\documentclass{IEEEtran}
\usepackage{graphicx} % Required for inserting images

\usepackage{amsmath}
\usepackage{amssymb, amsthm}
\usepackage{xcolor}
\usepackage{gensymb}

\usepackage{booktabs}   % \toprule, \midrule, \bottomrule 
\usepackage{multirow}

\newtheorem{theorem}{Theorem}
\newtheorem{example}{Example}

\title{Lattice Modulo Sampling
}

\author{Yhonatan Kvich, ~\IEEEmembership{Graduate Student Member,~IEEE,} Yonina C. Eldar,~\IEEEmembership{Fellow,~IEEE}\\Faculty of Math and Computer Science, Weizmann Institute of Science, Israel
\thanks{This research was supported by the Tom and Mary Beck Center for Renewable Energy as part of the Institute for Environmental Sustainability (IES) at the Weizmann Institute of Science, by the European Research Council (ERC) under the European Union’s Horizon 2020 research and innovation program (grant No. 101000967 and 101119062) and by the Israel Science Foundation (grant No. 536/22).}}

\begin{document}

\maketitle

\begin{abstract}
We propose a lattice-theoretic framework for modulo sampling of multidimensional bandlimited signals. Standard modulo analog-to-digital converters (ADCs) fold the signal component-wise into a square domain, reducing the recovery problem to independent one-dimensional problems.
We extend the recovery guarantees to any lattice, requiring the same sampling rate as in the standard component-wise modulo setting.
We also extend existing recovery algorithms to the general high-dimensional lattice setting.
Selecting a lattice with a smaller normalized second moment reduces the reconstruction mean squared error (MSE) through two complementary mechanisms: it lowers the folded signal power, which reduces the absolute noise energy at a fixed signal-to-noise ratio (SNR), and it reduces the 
quantization error when a matched lattice quantizer is applied.
Higher-dimensional lattices offer better second moment compared to the hypercube lattice, with gains that grow substantially with dimension.
Instantiating the framework in two dimensions with the hexagonal lattice reduces the MSE relative to the square at the same inradius by $\approx16.7\%$.
Furthermore, simulations on 8-dimensional signals using the $E_8$ lattice to achive $\approx 57\%$ in both additive and quantization noise.
A topological interpretation connects each folding geometry to a surface whose genus reflects the lattice complexity, and reveals a natural hardware implementation via comparator circuits.
\end{abstract}


\section{Introduction}

Analog-to-digital converters (ADCs) are essential components in 
digital signal processing systems, transforming analog signals into 
digital form. A key design challenge is accommodating high dynamic 
range (DR) input signals without clipping. Traditional approaches 
include oversampling \cite{marks1983restoring}, spectral gap 
exploitation \cite{abel1991restoring}, and automatic gain control 
\cite{perez2011automatic}, each carrying limitations in sampling 
rate, spectral knowledge, or nonlinear distortion.

An alternative strategy applies a modulo operation to the input 
signal before sampling, wrapping the amplitude within a fixed DR 
and avoiding saturation entirely \cite{park2007wide}. Bhandari 
et al.\ \cite{bhandari2020unlimited} showed that for bandlimited 
(BL) signals, sampling above the Nyquist rate enables recovery 
from modulo samples via higher-order differences (HOD), without 
any side information. This framework has since been extended to 
finite-rate-of-innovation (FRI) signals \cite{mulleti2024modulo}, 
shift-invariant (SI) spaces \cite{kvich2024Modulo, kvich2024Modulo2, 
kvich2025modulo}, sparse signals \cite{prasanna2020identifiability},
wavelet-based representations \cite{rudresh2018wavelet}, and 
direction-of-arrival estimation \cite{zhang2024line}, and validated 
on hardware prototypes 
\cite{bhandari2021unlimited, mulleti2023hardware, kvich2025practical}.
Several recovery algorithms have been developed to improve upon 
HOD, which is noise-sensitive and requires high oversampling factor ($\mathrm{OF}$). The 
$B^2R^2$ algorithm \cite{azar2022residual} recovers the signal 
using time- and frequency-domain separation, while 
LASSO-$B^2R^2$ \cite{shah2023lasso} exploits the sparsity of 
the residual differences to reduce computational cost and improve 
noise robustness. Deep unfolding approaches 
\cite{kvich2025msquid} further improve performance at low 
$\mathrm{OF}$ by learning recovery parameters from data.

The work of \cite{bouis2021multidimensional} considered multidimensional 
modulo sampling, where the signal maps from a higher-dimensional domain 
to the real line.
Existing works on 2D and complex-valued signals apply the modulo 
operator independently to the real and imaginary components, 
reducing the problem to two parallel 1D recoveries 
\cite{florescu2022surprising}. While straightforward, this 
component-wise approach does not exploit the joint geometry of the 
2D folding domain to improve robustness to noise and quantization.

In this paper, we propose a lattice-theoretic framework for modulo sampling of multidimensional BL signals, replacing the standard component-wise square folding with a general lattice modulo operator. We extend the recovery guarantees to any lattice, requiring the same OF as in the standard component-wise modulo setting, and show that existing recovery algorithms naturally extend to the general lattice setting. We analyze the reconstruction mean squared error (MSE) through the normalized second moment and show that selecting a lattice with a smaller second moment reduces the MSE through two complementary mechanisms: it lowers the folded signal power, reducing the absolute noise energy at a fixed signal-to-noise ratio (SNR), and when the quantizer is matched to the lattice by placing quantization points on a scaled copy of the lattice itself, it reduces the quantization error. Higher-dimensional lattices can offer strictly smaller second moments than the hypercube, with gains that grow with dimension. As a concrete two-dimensional instantiation, the hexagonal lattice reduces the MSE by $\approx 16.7\%$ relative to the square at the same inradius, and simulations on 8-dimensional signals using the $E_8$ lattice confirm a $\approx 57\%$ reduction under both additive and quantization noise. Finally, we provide a topological interpretation of the modulo operator: each lattice defines a flat manifold whose facet structure determines a theoretical hardware implementation in which each pair of opposite facets is monitored by two comparators that inject a correction along the corresponding vector upon crossing.

The remainder of this paper is organized as follows. Section~\ref{sec:mod_high_dim} introduces the general lattice modulo operator and establishes the unique identification guarantee. Section~\ref{sec:recovery} presents practical recovery algorithms, extending existing methods to the general lattice setting. Section~\ref{sec:noise_analysis} analyzes the robustness of the framework to additive noise and quantization, with the performance characterized by geometric properties of the lattice Voronoi cell. Section~\ref{sec:topology} provides a topological interpretation of the modulo operator and describes a theoretical hardware architecture based on comparator circuits. Section~\ref{sec:simulations} presents simulations under both additive noise and quantization, validating the theoretical predictions and Section~\ref{sec:conclusion} concludes 
the paper.

% In this paper, we propose a lattice-theoretic framework for modulo sampling that extends from one dimension to higher dimensions. We define a general lattice modulo operator and show that existing recovery algorithms naturally extend to any lattice. We analyze the quantization error through the normalized second moment and show that the mean squared error (MSE) is determined jointly by the lattice shape and cell volume, with reductions that grow with dimension through the use of known optimal lattices. As a concrete instantiation, we consider the two-dimensional case and show that the hexagonal lattice, which minimizes the normalized second moment among 2D lattice tilings, achieves a lower MSE than the square at the same threshold. We further demonstrate through simulations on 8-dimensional signals that the gains grow with dimension: the $E_8$ lattice reduces the MSE by $\approx57\%$ relative to the hypercube at the same threshold. We also provide a topological interpretation connecting each folding geometry to a surface, linking the lattice structure to a theoretical hardware implementation via comparator circuits.

\section{Modulo for Higher Dimensions}
\label{sec:mod_high_dim}

\subsection{Background}

The modulo operator with threshold $\lambda > 0$ is defined as
\begin{equation}
    \mathcal{M}_\lambda(x) = \left((x + \lambda) \bmod 2\lambda\right) - \lambda 
    \,:\, \mathbb{R} \to [-\lambda, \lambda]
\end{equation}
mapping any real-valued signal into the bounded interval $[-\lambda, \lambda]$ 
regardless of its DR. This operator serves as a general-purpose 
front-end for high DR acquisition.

For a BL signal $f(t) \in PW_\Omega$ (the Paley-Wiener space 
of functions with Fourier support in $[-\Omega, \Omega]$), the folded samples are 
$y_k = \mathcal{M}_\lambda f[kT]$. It has been shown that sampling above the 
Nyquist rate is sufficient to uniquely identify $f(t)$ from its modulo samples 
alone \cite{azar2025unlimited}. The core recovery challenge is identifying, for each sample, the 
integer $p_k \in \mathbb{Z}$ such that
\begin{equation}
    f[kT] = y_k + 2\lambda p_k .
\end{equation}

Several algorithms address this recovery problem, including HOD 
\cite{bhandari2020unlimited}, $B^2R^2$ \cite{azar2022residual}, and 
LASSO-$B^2R^2$ \cite{shah2023lasso}, all of which reduce to identifying 
the correct integer $p_k$ at each sample. This becomes increasingly 
challenging at smaller $\lambda$ or lower $\mathrm{OF}$, as the candidates are closer together, and noise further complicates the disambiguation between adjacent values.

%%%%%%%%%%%%%%%%%%%%%%%%%%%%

For 2D signals or complex-valued signals 
$f(t) = f_I(t) + j f_Q(t) \in \mathbb{C}$, the natural extension applies 
the modulo operator independently to each component:
\begin{equation}
    \mathcal{M}_\lambda^{\mathrm{2D}}(f) 
    = \mathcal{M}_\lambda(f_I) + j\,\mathcal{M}_\lambda(f_Q)
\end{equation}
confining the signal to the square $[-\lambda,\lambda]^2 \subset \mathbb{C}$. 
Recovery then reduces to two parallel 1D problems on the real and imaginary 
channels independently.

While conceptually straightforward, this component-wise approach overlooks the 
joint geometry of the 2D folding domain. A fundamental observation motivates our 
work: whereas in 1D the admissible interval $[-\lambda,\lambda]$ scales linearly 
with $\lambda$, in 2D the admissible region $[-\lambda,\lambda]^2$ has area 
$4\lambda^2$, scaling quadratically. The choice of folding geometry, that is, 
which shape tiles the complex plane, therefore directly governs the quantization 
error, the recovery noise tolerance, and the hardware complexity, in ways that 
have no 1D counterpart. 
This motivates the lattice-theoretic analysis of the following section.

\subsection{Lattice Modulo Folding}

The component-wise operator $\mathcal{M}_\lambda^{\mathrm{2D}}$ confines the 
signal to the square $[-\lambda,\lambda]^2$, which is the square lattice centered at the origin. This is a special case of a more general construction. A lattice $\Lambda \subset \mathbb{R}^n$ is the set of all integer 
linear combinations of $n$ linearly independent basis vectors 
$\{v_1, \ldots, v_n\} \subset \mathbb{R}^n$,
\begin{equation}
    \Lambda = \left\{ \sum_{i=1}^n k_i v_i \;:\; k_i \in \mathbb{Z} \right\} 
    = \mathbf{B}\mathbb{Z}^n
\end{equation}
where $\mathbf{B} = [v_1 \mid \cdots \mid v_n]$ is the generator matrix.
The Voronoi cell $\mathcal{V}$ of the origin is the set of points in $\mathbb{R}^n$ that are at least as close to the origin as to any other lattice point.
Denote the translates $\{\mathcal{V} + p\}_{p \in \Lambda}$ tile $\mathbb{R}^n$ without overlap, covering the entire space with shifted copies of $\mathcal{V}$.
The cell volume $V:=\mathrm{vol}(\mathcal{V}) = |\det \mathbf{B}|$. The group structure of 
$\Lambda$ guarantees that all Voronoi cells are congruent translates of 
$\mathcal{V}$, and that there exists a minimum distance 
$d_{\min} = \min_{p \in \Lambda \setminus \{0\}} \|p\|$ between any two 
distinct lattice points.

The modulo operator associated with $\Lambda$ maps any $\mathbf{z} \in \mathbb{R}^n$ 
to its residue within $\mathcal{V}$,
\begin{equation}
    \mathcal{M}_\Lambda(\mathbf{z}) = \mathbf{z} - \mathcal{Q}_{\Lambda}(\mathbf{z})
\end{equation}
where $\mathcal{Q}_{\Lambda}(\mathbf{z}) = \arg\min_{\mathbf{p} \in \Lambda} \|\mathbf{z} - \mathbf{p}\|_2$ is the 
nearest lattice point to $\mathbf{z}$. Therefore
$\mathcal{M}_\Lambda:\mathbb{R}^n\to \mathcal{V}$ and the samples of the folded signal 
$y[k] = \mathcal{M}_\Lambda(f[kT])$ lie in $\mathcal{V}$ for all $k$. The matrix $\mathbf{y}$ will denote the high-dimensional samples with time as the first coordinate.
We now state the identification guarantee for high-dimensional BL signals sampled under 
the lattice modulo operator. Denote $PW_\Omega^{\mathbb{R}^n}$ as a high-dimensional BL space with band $[-\Omega, \Omega]$ in each coordinate. The needed sampling rate will be strictly above Nyquist, meaning its the same rate for all lattices including the existing square lattice.

\begin{theorem}[Lattice Unlimited Sampling]
Let $f(t) \in PW_\Omega^{\mathbb{R}^n}$ and let 
$y[k] = \mathcal{M}_\Lambda f[kT]$ denote its lattice modulo 
samples at rate above Nyquist $T<\frac{\pi}{\Omega}$. Then $f(t)$ can be uniquely identified from the samples $y[k]$.
\end{theorem}

\begin{proof}
Since $f(t)$ have finite energy, we know from Riemann-Lebesgue Lemma that $f(t)\to 0$ as $|t|\to\infty$. Therefore there exist $N$, such that $|f[kT]|<\frac{d_{\min}}{2}$ for $|k|>N$. So in this case, $y[k] =\mathcal{M}_\Lambda  f [kT] = f [kT]$. Using \cite{kvich2025modulo} we know that a BL signal can be uniquely identifies from its tail sampled strictly above its Nyquist. Applying this to each coordinate of $f(t)$ concludes the proof.
\end{proof}

The theorem above guarantees the existence of a unique BL signal consistent with the modulo samples, but does not provide a practical recovery procedure. The proof exploits tail samples where the signal lies within $\mathcal{V}$ and therefore coincides with the folded signal, but such samples carry little energy and are difficult to identify reliably, particularly under noise or quantization. The following section presents practical recovery algorithms, extending existing methods to general higher-dimensional lattices.

\section{Lattice Modulo Recovery}
\label{sec:recovery}

Recovery of $f[kT_s]$ from $y[k]$ requires identifying the lattice point 
$p[k] \in \Lambda$ such that $f[kT_s] = y[k] + p[k]$. 
Existing recovery algorithms can be extend to the general 
lattice setting with the major change in the rounding step.
The HOD algorithm computes higher-order differences of the folded 
samples $y[k]$. At an oversampling factor of at least $2\pi e$ 
times the Nyquist rate \cite{bhandari2017unlimited}, it is guaranteed that for a 
sufficiently high difference order $N$, the $N$-th order differences 
of the true samples satisfy $|\Delta^N f[kT_s]| < d_{\min}/2$, where 
$d_{\min}$ is the minimum distance of $\Lambda$. Since the modulo 
operation does not affect values smaller than $d_{\min}/2$, the 
$N$-th order differences of $y[k]$ and $f[kT_s]$ coincide, and 
rounding $\Delta^N \mathbf{y}$ to the nearest point in $\Lambda$ yields 
the correct fold offset at each step. The original samples are then 
reconstructed by applying the corresponding $N$-th order 
summation to the recovered differences.

The $B^2R^2$ algorithm exploits two separation principles. In the 
frequency domain, since $f(t)$ is BL, the out-of-band discrete time Fourier transform (DTFT) 
components of the true samples vanish, so the out-of-band DTFT of $\mathbf{y}$ equals that of $-\mathbf{p}$.
Denoting by $\mathbf{F} \in \mathbb{C}^{N_{\mathrm{oob}} \times K}$ 
the matrix whose $(m,k)$-th entry is
\begin{equation}
    [\mathbf{F}]_{m,k} = e^{-j\Omega_m k}, 
    \quad \Omega_m \notin [-\Omega_{\max}, \Omega_{\max}]
\end{equation}
where $\Omega_m$ are the out-of-band normalized digital frequencies 
and $K$ is the number of samples. The out-of-band constraint reads 
$ \mathbf{F} \mathbf{y} = - \mathbf{F} \mathbf{p}$.
In the time domain, $p[k] \in \Lambda$ has restricted support, meaning $p[k] = 0$ for $|k|>K_{\max}$. This restricts the search space, denote as $\mathbf{S}\subset\mathbb{R}^{K\times n}$. 
$B^2R^2$ combines these two constraints to recover $\mathbf{p}$ from the samples $\mathbf{y}$.
The resulting least-squares problem
\begin{equation}
    \hat{\mathbf{p}} = \arg\min_{\mathbf{p} \in \mathbf{S}}
    \left\| \mathbf{F} \mathbf{p}  
    + \mathbf{F} \mathbf{y} \right\|_2^2 .
\end{equation}
This is solved via standard gradient descent. The final reconstruction is $\hat{f}[kT_s] = \mathbf{y} + \mathcal{Q}_\Lambda (\hat{\mathbf{p}})$, where $\mathcal{Q}_\Lambda$ is applied across the sampling dimension.
LASSO-$B^2R^2$ exploits the same out-of-band equality but recovers 
the full residual globally by noting that the first difference 
$\Delta \mathbf{p}$ is sparse, being nonzero only at fold events. Letting 
$\mathbf{v} = \Delta \mathbf{p}$ and $\mathbf{C}$ denote the cumulative 
summation matrix on the first dimension, so that $\mathbf{p} = \mathbf{C}\mathbf{v}$. The 
$\ell_1$ constraint on $\mathbf{v}$ is yields the convex problem
\begin{equation}
    \hat{\mathbf{v}} = \arg\min_{\mathbf{v} \in \mathbb{R}^{K\times n}}
    \left\|  \mathbf{F}\mathbf{C} \mathbf{v} 
    + \mathbf{F} \mathbf{y} \right\|_2^2 
    + \mu \|\mathbf{v}\|_1
\end{equation}
where $\mu$ is a regularization factor and the task can be solved via proximal gradient descent. The signal is 
then reconstructed as $\hat{f}[kT_s] = \mathbf{y} +\mathcal{Q}_\Lambda ( \mathbf{C}\hat{\mathbf{v}})$, 
which include a rounding step onto $\Lambda$.

Finding the nearest lattice point $\mathcal{Q}_\Lambda(\mathbf{x})$ is required at the final rounding step of each recovery algorithm, and also plays a central role in the quantization analysis presented in the next section. For the standard lattices used in this paper, efficient exact algorithms exist \cite{conway1982fast}. We summarize the relevant cases here.

For the integer lattice $\mathbb{Z}^n$, the nearest point is obtained by rounding each coordinate independently, $\mathcal{Q}_{\mathbb{Z}^n}(\mathbf{x}) = \mathrm{round}(\mathbf{x})$, where ties are broken by choosing the integer with smaller absolute value.
The lattice $D_n$ is defined as the set of integer vectors with even coordinate sum,
\begin{equation}
    D_n = \left\{ \mathbf{x} \in \mathbb{Z}^n : 
    \sum_{i=1}^n x_i \equiv 0 \pmod{2} \right\}.
\end{equation}
Following \cite{conway1982fast}, define $f(\mathbf{x}) = \mathrm{round}(\mathbf{x})$ as the component-wise nearest integer, and $g(\mathbf{x})$ as the same rounding except that the component furthest from an integer is rounded in the opposite direction (with ties broken by lowest index). Since $f(\mathbf{x})$ and $g(\mathbf{x})$ differ in exactly one coordinate by $\pm 1$, their coordinate sums have opposite parities: exactly one of the two lies in $D_n$. The nearest point in $D_n$ is therefore the one among $f(\mathbf{x})$ and $g(\mathbf{x})$ whose coordinate sum is even. If $\mathbf{x}$ is equidistant from two or more points of $D_n$, this procedure returns the nearest point with smallest norm. The algorithm is exact, and runs in $O(n)$.

The special case of $E_8$ lattice can be seen as the union
\begin{equation}
    E_8 = D_8 \cup \left(D_8 + 
    \tfrac{1}{2}\mathbf{1}\right)
\end{equation}
where $\mathbf{1} = (1,\ldots,1)^\top$ and $D_8 + \frac{1}{2}\mathbf{1}$ is the coset of $D_8$ obtained by a half-integer shift in all coordinates. The nearest point in $E_8$ is found by applying  $\mathcal{Q}_{D_8}$ twice: once to $\mathbf{x}$ directly, and once to the shifted input $\mathbf{x} - 
\frac{1}{2}\mathbf{1}$ (then shifting the result back), and returning the closer of the two candidates.
Meaning

\begin{equation}
\begin{aligned}
    \mathbf{c}_1 &= \mathcal{Q}_{D_8}(\mathbf{x}) \\
    \mathbf{c}_2 &= \mathcal{Q}_{D_8}\!\left(\mathbf{x} - 
    \tfrac{1}{2}\mathbf{1}\right) + \tfrac{1}{2}\mathbf{1} \\
    \mathcal{Q}_{E_8}(\mathbf{x}) &= 
    \begin{cases} 
        \mathbf{c}_1 & \text{if } \|\mathbf{x} - \mathbf{c}_1\|_2 
        \leq \|\mathbf{x} - \mathbf{c}_2\|_2 \\
        \mathbf{c}_2 & \text{otherwise.}
    \end{cases}
\end{aligned}
\end{equation}

\begin{example}
Consider a signal sample $\mathbf{x} = (2.3, -3.1, 5.6, 1.2, -4.4, 3.1, 6.7, -2.2)$ and we will compute  $\mathcal{Q}_{E_8}(\mathbf{x})$ via the two-coset algorithm.

\textit{Coset 1.} Apply $\mathcal{Q}_{D_8}$ to $\mathbf{x}$.
Rounding component-wise gives $f(\mathbf{x}) = (2, -3, 6, 1, -4, 3, 7, -2)$ with coordinate sum $10$, which is even, so $\mathcal{Q}_{D_8}(\mathbf{x}) = f(\mathbf{x})$ with distance $d_1 \approx 0.775$.

\textit{Coset 2.} Shift by $-\frac{1}{2}\mathbf{1}$ to get $\mathbf{x} - \frac{1}{2}\mathbf{1} = (1.8, -3.6, 5.1, 0.7, -4.9, 2.6, 6.2, -2.7)$. Rounding gives $f(\mathbf{x}) = (2, -4, 5, 1, -5, 3, 6, -3)$ with coordinate sum $5$, which is odd. The component 
furthest from an integer is the second entry (tied with the sixth), so $g(\cdot)$ rounds it the other way: $g = (2, -3, 5, 1, -5, 3, 6, -3)$ with sum $6$, which is even. Shifting back gives the 
second candidate $(2.5, -2.5, 5.5, 1.5, -4.5, 3.5, 6.5, -2.5)$ 
with distance $d_2 \approx 0.894$.

\textit{Selection.} Since $d_1 < d_2$, the nearest point is 
$\mathcal{Q}_{E_8}(\mathbf{x}) = (2, -3, 6, 1, -4, 3, 7, -2)$ 
and the folded sample is
\begin{equation}
\begin{aligned}
    \mathcal{M}_{E_8}(\mathbf{x}) = \mathbf{x} - \mathcal{Q}_{E_8}(\mathbf{x}) 
    = \\(0.3, -0.1, -0.4, 0.2, -0.4, 0.1, -0.3, -0.2)
\end{aligned}
\end{equation}
which lies within the $E_8$ Voronoi cell.
\end{example}

\section{Noise and Quantization}
\label{sec:noise_analysis}

The recovery algorithms presented in Section~\ref{sec:recovery} operate on the folded samples $y[k] = \mathcal{M}_\Lambda f[kT]$ as received by the ADC. The goal is to find the modulo effect $p[k]=f[kT]-y[k]\in\Lambda$. In practice, the ADC can introduce distortion of one of two kinds. In the additive noise setting, the received sample is $\tilde{\mathbf{y}} = \mathbf{y} + \mathbf{n}$ where $\mathbf{n}$ is a noise vector. In the quantization setting, the ADC maps each folded sample to the nearest point on a quantization grid. We consider two architectures: a standard component-wise scalar quantizer, independent of the lattice used; and a matched lattice quantizer that projects onto the scaled lattice $2^{-B}\Lambda$, giving $\tilde{\mathbf{y}} = \mathcal{Q}_{2^{-B}\Lambda}(\mathbf{y})$ with $B$ denoting the number of bits.

In both settings, the distortion $\mathbf{e} = \tilde{\mathbf{y}} - \mathbf{y}$ determines the reconstruction quality through two mechanisms: a smaller $\|\mathbf{e}\|$ reduces the probability that the unfolding step fails. In addition, when unfolding succeeds (meaning $\hat{\mathbf{p}}=\mathbf{p}$) the residual MSE equals

\begin{equation}
\left\|\hat{f} - f\right\| = \left\| \tilde{\mathbf{y}} - \mathbf{y}\right\|=\|\mathbf{e}\|.
\end{equation}
As we show in the remainder of this section, both mechanisms are governed by the geometry of the lattice Voronoi cell, and selecting a better lattice reduces the distortion — and therefore the reconstruction error — simultaneously in both regimes.

Given an $n$-dimensional lattice $\Lambda$ and its Voronoi cell $\mathcal{V}$ with volume $V$. The normalized second moment, also 
called the \emph{dimensionless second moment} \cite{conway1984voronoi, conway2003voronoi, conway2013sphere}, is
\begin{equation}
    G_\Lambda = \frac{1}{n \cdot V^{1+2/n}} \int_{\mathcal{V}} \|r\|^2 \, \mathrm{d}r
\end{equation}
and the mean squared quantization error per dimension for a uniform input is 
$\mathrm{MSE} = n \cdot G \cdot V^{2/n}$.

The natural comparison between two lattices is at equal cell area, scaled by $\lambda^n$. At fixed area, the MSE ratio between two lattices is simply the ratio between the second moments at equal inradius $\lambda$.
When the cell volumes differ, giving
\begin{equation}
    \frac{\mathrm{MSE}_{\Lambda_1}}{\mathrm{MSE}_{\Lambda_2}} 
    = \frac{G_{\Lambda_1} \cdot V_{\Lambda_1}^{2/n}}
           {G_{\Lambda_2} \cdot V_{\Lambda_2}^{2/n}}
\end{equation}
where both volumes scale as $\lambda^n$.
The choice of lattice affects the recovered signal quality through two distinct mechanisms.

\textbf{Additive noise.} Assuming the folded signal $ \mathcal{M}_\Lambda(f)$ is uniformly distributed over $\mathcal{V}$, its average power 
equals $\mathrm{MSE}_\Lambda$. Since the SNR is defined as the ratio of signal power to noise power, a smaller $\mathrm{MSE}_\Lambda$ at the same inradius $\lambda$ directly implies lower absolute noise power. The benefit of one lattice over another is captured by the ratio $\mathrm{MSE}_{\Lambda_1}/\mathrm{MSE}_{\Lambda_2}$, which reduces the probability of a fold identification error and the residual error in cases of perfect recovery.

\textbf{Quantization.} When the folded samples are passed through a $B$-bit ADC, two quantization architectures are possible. The first applies a standard component-wise scalar quantizer to the folded signal regardless of the lattice geometry, which is attractive since it requires no modification to conventional ADC hardware. The second exploits the lattice structure by quantizing onto the scaled lattice $2^{-B}\Lambda$, which is a finer version of $\Lambda$ with $2^{nB}$ points per unit cell. Since all Voronoi cells of $2^{-B}\Lambda$ are congruent scaled copies of $\mathcal{V}$, the mean squared 
quantization error is
\begin{equation}
    \mathrm{MSE}_{\Lambda}^{\mathrm{quant}} = n \cdot G \cdot 
    \left(\frac{V}{2^{nB}}\right)^{2/n}
    = \mathrm{MSE}_\Lambda \cdot 4^{-B}.
\end{equation}
The ratio of quantization errors between two lattices with $B$ bits is therefore $\mathrm{MSE}_{\Lambda_1}^{\mathrm{quant}} / \mathrm{MSE}_{\Lambda_2}^{\mathrm{quant}} = \mathrm{MSE}_{\Lambda_1}/\mathrm{MSE}_{\Lambda_2}$. Note that this ratios is independent of the number of bits $B$. The gain from choosing an optimal lattice is the same at every bit depth, as confirmed in 
Table~\ref{tab:lattice_G}.

For the hypercube lattice, straightforward integration gives $G_\square = 1/12$ 
for all $n$. In two dimensions, the hexagonal lattice $A_2$, generated by the 
basis matrix
\(
    \mathbf{B}_{A_2} = 2\lambda 
    \begin{pmatrix} 1 & 1/2 \\ 0 & \sqrt{3}/2 \end{pmatrix}
\)
is known to minimize $G$ among all 2D lattice tilings, with 
$G_{A_2} = 5/(36\sqrt{3}) \approx 0.0802$. At equal inradius, 
$V_{A_2} = 2\sqrt{3}\lambda^2$ and $V_\square = 4\lambda^2$, giving
\begin{equation}
    \frac{\mathrm{MSE}_{A_2}}{\mathrm{MSE}_{\square}} 
    = \frac{G_{A_2} \cdot V_{A_2}}{G_{\square} \cdot V_{\square}}
    = \frac{5/(36\sqrt{3}) \cdot 2\sqrt{3}\lambda^2}{(1/12) \cdot 4\lambda^2}
    = \frac{5}{6} \approx 0.833
\end{equation}
a reduction of $16.7\%$ in mean square quantization error compare to regular sampling.

The gain increases with dimension. One approach is to couple signal dimensions 
in pairs and apply the $A_2$ lattice to each pair independently, yielding 
$16.7\%$ MSE reduction per pair with compared with regular sampling. More generally, higher-dimensional lattices 
can be constructed by combining lower-dimensional ones, or optimized directly 
for the specific dimension \cite{agrell2002optimization}.
Table~\ref{tab:lattice_G} reports known 
lattice quantizers for selected dimensions and their MSE relative to the hypercube at equal inradius.

The MSE reduction offered by a lattice can be interpreted as an equivalent improvement in sampling rate, SNR, or bit depth, providing intuition for its practical value. Since the reconstruction MSE under additive noise scales as $\mathrm{OF}^{-1}$ and as $4^{-B}$ under quantization, a reduction in MSE by a factor $r$ is equivalent to multiplying the oversampling factor by $1/r$, increasing the SNR by $10\log_{10}(1/r)$ dB, or saving $\log_4(r)$ bits of ADC resolution. 
For the $E_8$ lattice in eight dimensions, the $57\%$ MSE reduction ($-3.67$ dB) is equivalent to a $2.3\times$ increase in oversampling factor, a 
$3.67$ dB SNR improvement, or approximately $0.6$ saved bits. For the Leech lattice $\Lambda_{24}$ in 24 dimensions, the $80.3\%$ reduction ($-7.06$ dB) corresponds to a $5\times$ increase in oversampling factor, a $7.06$ dB SNR improvement, or approximately $1.2$ saved bits.

An important consequence concerns the two-lattice architecture in the quantization setting, where an optimal lattice modulo operator is paired with a standard component-wise scalar quantizer. Note that in the additive noise setting this architecture still benefits from the reduced folded signal power; the following argument applies specifically to 
quantization. The scalar quantizer partitions $\mathbb{R}^n$ into hypercubes of side $\delta = 2\lambda/2^B$ independently of the folding lattice, so the quantization error $\mathbf{e}$ is uniformly distributed over $[-\delta/2, \delta/2]^n$ regardless of which Voronoi cell $\mathbf{y}$ came from. The quantization MSE is therefore identical for all folding lattices. Near the cell boundary, a quantization bin may straddle $\partial\mathcal{V}$ and map $\mathbf{y}$ outside $\mathcal{V}$, but this probability scales as $O(\delta/V^{1/n})$ and vanishes as $B$ increases — the same regime required for successful modulo recovery. We conclude that realizing the full quantization benefit of an optimal lattice requires a matched quantizer; a mismatched scalar quantizer yields the same MSE as the standard square modulo regardless of the folding lattice.

\begin{table}[t]
\caption{MSE relative to the hypercube at equal inradius $\lambda$ for 
known lattice quantizers.}
\label{tab:lattice_G}
\centering
\begin{tabular}{clcccc}
\toprule
$n$ & Lattice & $G_\Lambda$ & 
$V_\Lambda / (2\lambda)^n$ & 
$\mathrm{MSE}_\Lambda / \mathrm{MSE}_\square$ & Reduction \\
\midrule
$1$  & $\mathbb{Z}$   & $0.0833$ & $1.000$ & $1.000$ & $0\%$    \\
$2$  & $A_2$          & $0.0802$ & $0.866$ & $0.833$ & $16.7\%$ \\
$3$  & $A_3^*$       & $0.0785$ & $0.707$ & $0.748$ & $25.2\%$ \\
$4$  & $D_4$          & $0.0766$ & $0.500$ & $0.650$ & $35.0\%$ \\
$8$  & $E_8$          & $0.0717$ & $0.063$ & $0.430$ & $57.0\%$ \\
$24$ & $\Lambda_{24}$ & $0.0658$ & $5.96\times 10^{-8}$ & $0.197$ & $80.3\%$ \\
\bottomrule
\end{tabular}
\end{table}

The problem of minimizing the normalized second moment is related to, but distinct from, the sphere packing problem; both favor dense, ``round" Voronoi cells but the two optima do not coincide in general. For sphere packing, optimality is proven for $n \leq 8$ by classical results \cite{conway2013sphere}. For lattice quantization, optimality is proven only for $n \leq 3$, with the provably optimal lattices used in Table~\ref{tab:lattice_G} for these dimensions \cite{conway2013sphere}. For $n \geq 4$, the optimal lattice quantizers are not known; the lattices in Table~\ref{tab:lattice_G} are the best known from \cite{conway2013sphere} and are conjectured to be optimal \cite{conway1985lower}. The same work establishes a lower bound on the achievable normalized second moment, and the best known lattices are close to this bound, so any improvement from an as-yet-undiscovered optimal lattice would be small and of limited practical significance.

As the dimension grows, the dominant contribution to the MSE reduction comes from the volume ratio $V_\Lambda / V_\square$: optimal lattices pack space so efficiently that their Voronoi cells become exponentially smaller than the hypercube at the same inradius. Zador's asymptotic bound \cite{conway2013sphere} gives a limiting normalized second moment of $G \to 1/(2\pi e) \approx 0.0585$. Combined with the vanishing volume ratio, the MSE relative to the hypercube approaches zero as $n \to \infty$, meaning that in sufficiently high dimensions an optimal lattice quantizer can reduce the quantization error to an arbitrarily small fraction of that of the hypercube at the same inradius.

\section{Topology and Hardware Implementation}
\label{sec:topology}

The modulo operator admits a natural topological interpretation: identifying 
opposite faces of the fundamental domain corresponds to gluing those faces 
together, and each pair of glued faces requires one folding loop implemented 
by two comparators in hardware. In 1D, identifying the endpoints $\pm\lambda$ 
of $[-\lambda,\lambda]$ yields a circle $S^1$, and the hardware implementation 
injects $\mp 2\lambda$ whenever the signal crosses $\pm\lambda$ 
\cite{bhandari2021unlimited, mulleti2023hardware, kvich2025practical}, 
requiring two comparators. In 2D, the square domain $[-\lambda,\lambda]^2$ 
has two pairs of opposite edges; gluing left to right and top to bottom 
yields a torus $T^2 = S^1 \times S^1$, implemented with two independent 
folding loops for a total of four comparators.

The hexagonal Voronoi cell has three pairs of opposite edges, corresponding 
to the three lattice directions defined by $\mathbf{B}_{A_2}$. Gluing all 
three pairs of opposite edges yields a genus-2 surface, a closed surface with 
two ``holes", reflecting the additional topological complexity of the hexagonal 
tiling relative to the square. Each pair requires two comparators monitoring 
the thresholds defined by that edge pair and injecting a correction along the 
corresponding lattice direction upon crossing, giving three independent folding 
loops and six comparators in total.
The three edge pairs correspond to the following boundary conditions on the real signal $(x, y) \in \mathbb{R}^2$. The first pair monitors $x = \pm\lambda$ and injects $\mp(2\lambda, 0)$ along the horizontal axis. The second pair monitors 
$x + \sqrt{3}\,y = \pm 2\lambda$ and injects $\mp(\lambda, \lambda\sqrt{3})$ along the direction of $v_2$, at $60\degree$ to the horizontal. The 
third pair monitors $x - \sqrt{3}\,y = \pm 2\lambda$ and injects $\mp(\lambda, -\lambda\sqrt{3})$ along the direction of $v_1 - v_2$, at $-60\degree$. Each comparator evaluates a linear combination of the two signal channels and triggers a reset of $\mp 2\lambda$ in the corresponding lattice direction upon crossing.

In higher dimensions, the lattice Voronoi cell has an increasing number of facet pairs. For the $n$-dimensional hypercube there are $n$ pairs, requiring $2n$ comparators.
In general, each facet of the Voronoi cell is associated with a relevant vector $\mathbf{p}_i \in \Lambda \setminus\{\mathbf{0}\}$ — a lattice point whose perpendicular bisector with the origin defines that facet. The number of such vectors of a lattice is called the kissing number. The corresponding threshold condition is $\langle \mathbf{p}_i, \mathbf{x} \rangle = \pm\|\mathbf{p}_i\|^2_2/2$, a fixed linear condition on the signal coordinates determined entirely by $\mathbf{p}_i$. When the signal crosses this threshold, the comparator triggers a correction vector $\mp \mathbf{p}_i$. The threshold coefficients and the injection vectors are thus fully determined by the lattice $\Lambda$, requiring no additional design choices beyond the selection of $\Lambda$. The resulting hardware is a fixed analog circuit that operates without runtime reconfiguration.

For denser lattices such as $D_4$, $E_8$, or $\Lambda_{24}$, 
the Voronoi cell has more facets, and gluing the corresponding opposite faces 
produces surfaces of increasingly high ''holes" of higher dimensions, reflecting the richer connectivity 
of the lattice. The number of comparators is the kissing number of the lattice, which grows with dimension and lattice complexity.
For example, the $E_8$ lattice has 240 relevant vectors, all of norm $\sqrt{2}$ in the unit lattice, forming 120 pairs $\pm\mathbf{p}_i$. The relevant vectors comprise two sets: vectors with exactly two nonzero entries of $\pm 1$ in all sign and coordinate combinations, and 
vectors with all eight entries equal to $\pm\frac{1}{2}$ with an even number of minus signs \cite{conway2013sphere}. Each relevant vector defines a facet via $\langle \mathbf{p}_i, \mathbf{x} \rangle = \pm 1$, requiring 240 comparators in the theoretical hardware implementation.
The comparator count for the geometries considered in this paper is summarized 
in Table~\ref{tab:comparators}.

\begin{table}[t]
\caption{Topology and theoretical hardware cost for each folding geometry.}
\label{tab:comparators}
\centering
\begin{tabular}{lccc}
\toprule
Geometry & Topology & Folding pairs & Comparators \\
\midrule
1D interval               & $S^1$   & $1$ & $2$ \\
2D Square $[-\lambda,\lambda]^2$ & $T^2$   & $2$ & $4$ \\
Hexagon $A_2$             & genus-2 & $3$ & $6$ \\
$D_4$             &  & $12$ & $24$ \\
$E_8$             &  & $120$ & $240$ \\
$\Lambda_{24}$             &  & $98280$ & $196560$ \\
\bottomrule
\end{tabular}
\end{table}

\section{Simulations}
\label{sec:simulations}

We demonstrate hexagonal modulo sampling and recovery on a complex 
BL signal $z(t) \in PW_\Omega^{\mathbb{C}}$, constructed as 
a sum of $N_c = 14$ complex sinusoids with frequencies drawn uniformly 
from $[-\Omega_{\max}, \Omega_{\max}]$ with $\Omega_{\max} = 10$ Hz. The signal is 
scaled so that $\max(|z_I|, |z_Q|) = 3\lambda$, ensuring the amplitude 
exceeds the ADC threshold by a factor of three. We sample at $6\times$ 
the Nyquist rate ($f_s = 120$ Hz), apply the square and hexagonal modulo 
operators separately, and recover using a \(B^2R^2\)-based algorithm, where 
the final rounding step projects onto the nearest point in 
$2\lambda\mathbb{Z}^2$ or $A_2$, respectively.

Fig.~\ref{fig:sim} shows the trajectories (top row) and 
the two components (bottom two rows). The 
original signal (dark, dashed) traverses multiple cells, while the 
folded signal (solid grey) is confined to the fundamental domain. In 
both the square and hexagonal cases, the recovered signal (dashed, 
colored) overlays the original to within machine precision.

\begin{figure}[t]
    \centering
    \includegraphics[width=\columnwidth]{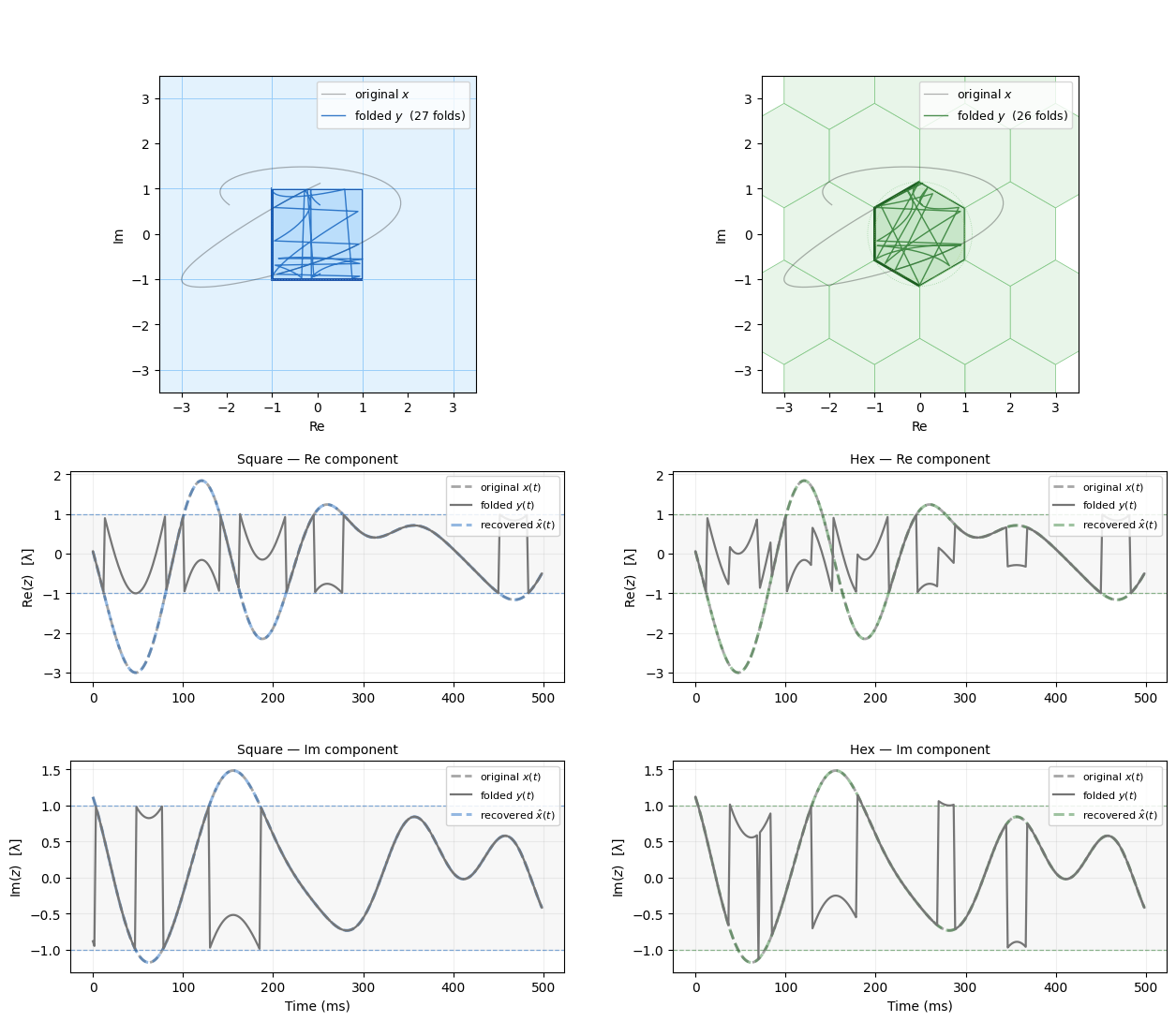}
    \caption{Square vs.\ hexagonal modulo sampling of a 2D BL 
    signal. Original 
    (dark dashed), folded (solid grey), recovered (colored dashed).}
    \label{fig:sim}
\end{figure}

To demonstrate that the MSE reduction grows with dimension, we 
extend the evaluation to 8-dimensional signals using the $E_8$ 
lattice.
We generate an 8-dimensional real-valued BL signal 
$f(t) \in \mathbb{R}^8$, where each channel is an 
independent sum of $N_c = 14$ sinusoids with frequencies drawn 
uniformly from $[-\Omega_{\max}, \Omega_{\max}]$, $\Omega_{\max} = 10$ Hz, 
normalised so that $\max_{t,d} |f_d(t)| = 1$, with modulo 
threshold $\lambda = 0.1$, corresponding to a dynamic range 
reduction factor of $1/\lambda = 10$.

Two noise models are evaluated, both at the same inradius $\lambda$. In the additive noise setting, zero-mean Gaussian noise is added to the folded samples at a prescribed SNR. In the quantization setting, three architectures are considered, differing in both the folding lattice and the quantizer applied after folding. The first applies component-wise square folding followed by a uniform mid-tread scalar quantizer with step size $\delta = 2\lambda/2^B$ per dimension (Sq+SqQ), serving as the baseline. The second applies $E_8$ folding followed by the same component-wise scalar quantizer ($E_8$+SqQ), isolating the contribution of the folding geometry while keeping the downstream ADC conventional. The third applies $E_8$ folding followed by quantization onto the scaled lattice $2^{-B}E_8$ ($E_8$+$E_8$Q), placing the quantization points on a finer copy of the $E_8$ lattice itself. Each quantization cell is a scaled copy of the $E_8$ Voronoi cell, achieving the full $G_{E_8}$ advantage over the hypercubic grid. The nearest-point projection onto $2^{-B}E_8$ is implemented via the exact algorithm of \cite{conway1982fast}.

Two noise models are evaluated. In the quantization setting, 
the folded samples are passed through a uniform mid-tread ADC 
with $B \in \{2, 4, 6, 8, 10, \infty\}$ bits, where $B = \infty$ 
denotes an ideal ADC. In the additive noise setting, uniform noise is added to the folded samples with $\mathrm{SNR} \in \{10, 15, 20, 25, 30, 35, \infty\}$ dB. 
Recovery uses HOD and $B^2R^2$, extended to 8 dimensions as described 
above, sweeping $\mathrm{OF} \in \{2, 4, 6, 8\}$ over 
$N = 500$ iid signals per case.

For the additive noise, failed in all tested 
configurations under noise and its results are omitted.
Table~\ref{tab:e8_recovery} shows the full fold recovery rate for 
$B^2R^2$. At low $\mathrm{OF} = 2$ all approaches fail 
for all finite SNR values. As OF increases, a clear recovery 
threshold emerges above which all approaches achieve perfect 
recovery. The $E_8$ approaches consistently reach this threshold 
at a lower SNR than the square, and wherever recovery succeeds, 
the $E_8$ lattice reduces the MSE by approximately $3.66$ dB, 
consistent with the theoretical prediction from 
Table~\ref{tab:lattice_G}.

\begin{table}[t]
\caption{Full fold recovery rate for $B^2R^2$ under additive noise,
$E_8$ lattice.}
\label{tab:e8_recovery}
\centering
\begin{tabular}{clccccccc}
\toprule
OF & Approach & 10 & 15 & 20 & 25 & 30 & 35 & $\infty$ \\
 & & dB & dB & dB & dB & dB & dB & \\
\midrule
\multirow{2}{*}{2}
 & Square & 0 & 0 & 0 & 0 & 0 & 0 & 1 \\
 & $E_8$ & 0 & 0 & 0 & 0 & 0 & 0 & 1 \\
\midrule
\multirow{2}{*}{4}
 & Square & 0 & 0 & 0 & 0 & 0.006 & 1 & 1 \\
 & $E_8$ & 0 & 0 & 0 & 0 & 0.472 & 1 & 1 \\
\midrule
\multirow{2}{*}{6}
 & Square & 0 & 0 & 0 & 0.996 & 1 & 1 & 1 \\
 & $E_8$ & 0 & 0 & 0 & 1 & 1 & 1 & 1 \\
\midrule
\multirow{2}{*}{8}
 & Square & 0 & 0 & 0.962 & 1 & 1 & 1 & 1 \\
 & $E_8$ & 0 & 0 & 1 & 1 & 1 & 1 & 1 \\
\bottomrule
\end{tabular}
\end{table}

Table~\ref{tab:e8_recovery_quant} shows the full fold recovery rate 
for $B^2R^2$ under quantization noise. The $E_8$ approaches reach the 
recovery threshold at a lower bit depth than the square: at 
$\mathrm{OF} = 6$, $E_8$+$E_8$Q achieves full recovery at $4$ bits while 
both Sq+SqQ and $E_8$+SqQ require $6$ bits.

A key distinction from the additive noise setting is that the 
folding geometry alone does not reduce the quantization MSE. 
Wherever recovery succeeds, $E_8$+SqQ achieves the same MSE as 
Sq+SqQ, since the square quantizer does not exploit the geometry 
of the $E_8$ Voronoi cell. By contrast, $E_8$+$E_8$Q achieves a 
consistent $3.66$ dB like in the additive case and as pridicted by the theory.
This confirms that realizing the full quantization advantage of 
$E_8$ requires both the matched folding operator and the matched 
quantizer.

\begin{table}[t]
\caption{Full fold recovery rate for $B^2R^2$ under quantization noise,
$E_8$ lattice.}
\label{tab:e8_recovery_quant}
\centering
\begin{tabular}{clccccccc}
\toprule
OF & Approach & 2b & 4b & 6b & 8b & 10b & $\infty$ \\
\midrule
\multirow{3}{*}{2}
 & Sq+SqQ & 0 & 0 & 0 & 0 & 1 & 1 \\
 & $E_8$+SqQ & 0 & 0 & 0 & 0 & 1 & 1 \\
 & $E_8$+$E_8$Q & 0 & 0 & 0 & 0.002 & 1 & 1 \\
\midrule
\multirow{3}{*}{4}
 & Sq+SqQ & 0 & 0 & 1 & 1 & 1 & 1 \\
 & $E_8$+SqQ & 0 & 0 & 1 & 1 & 1 & 1 \\
 & $E_8$+$E_8$Q & 0 & 0 & 1 & 1 & 1 & 1 \\
\midrule
\multirow{3}{*}{6}
 & Sq+SqQ & 0 & 0.956 & 1 & 1 & 1 & 1 \\
 & $E_8$+SqQ & 0 & 0.008 & 1 & 1 & 1 & 1 \\
 & $E_8$+$E_8$Q & 0 & 1 & 1 & 1 & 1 & 1 \\
\midrule
\multirow{3}{*}{8}
 & Sq+SqQ & 0 & 1 & 1 & 1 & 1 & 1 \\
 & $E_8$+SqQ & 0 & 1 & 1 & 1 & 1 & 1 \\
 & $E_8$+$E_8$Q & 0 & 1 & 1 & 1 & 1 & 1 \\
\bottomrule
\end{tabular}
\end{table}

\section{Conclusion}
\label{sec:conclusion}

We proposed a lattice-theoretic framework for modulo sampling of multidimensional BL signals, replacing the standard component-wise square folding with a general lattice modulo operator. We extended existing recovery algorithms to the general lattice setting, requiring the same oversampling rate. We showed that the reconstruction MSE is governed by the normalized second moment of the lattice Voronoi cell, and that a smaller second moment reduces the MSE through two complementary mechanisms: lower folded signal power at a fixed SNR, and lower quantization error when a matched lattice quantizer places points on a scaled copy of the lattice. Higher-dimensional lattices can offer strictly smaller second moments than the hypercube, with the MSE reduction reaching $\approx 57\%$ for $E_8$ in eight dimensions and $\approx 80\%$ for $\Lambda_{24}$ in twenty-four dimensions. As a two-dimensional instantiation, the hexagonal lattice achieves a MSE reduction at the same inradius. Monte Carlo simulations on 8-dimensional signals using the $E_8$ lattice validate the theoretical predictions under both additive and quantization noise. A topological interpretation connects each lattice to a flat manifold, with each pair of opposite Voronoi facets monitored by two comparators that inject a correction along the corresponding relevant vector, providing a theoretical hardware architecture that scales naturally with the lattice complexity.

\bibliographystyle{IEEEtran}
\bibliography{refs}

\end{document}